\documentclass[acmtog,authorversion]{acmart}
\usepackage[shortlabels]{enumitem}
\usepackage{color}
\AtBeginDocument{%
  }

\setcopyright{acmlicensed}
\copyrightyear{2026}
\acmYear{2026}
\acmVolume{45} 
\acmNumber{4} 
\acmArticle{49}
\acmConference[SIGGRAPH '26 Conference Papers]
  {Special Interest Group on Computer Graphics and Interactive Techniques Conference Conference Papers}
  {July 19--23, 2026}{Los Angeles, CA, USA}
\acmJournal{TOG}
\definecolor{AAA}{rgb}{1.0, 0.13, 0.32}
\definecolor{BBB}{rgb}{0.2, 0.1, 1}
\definecolor{CCC}{rgb}{0.0, 1, 0}
\definecolor{DDD}{rgb}{0.0, 0, 1}
\newcommand{\ul}[1]{\underline{#1}}

\usepackage{xcolor}
\usepackage{colortbl}
\usepackage{graphicx} % for resizebox

\definecolor{headerblue}{rgb}{0.8,0.87,0.94}
\definecolor{rowgray}{rgb}{0.95,0.95,0.95}
\definecolor{myorange}{HTML}{ff9a33}

\definecolor{onec}{HTML}{fe218b}
\definecolor{twoc}{HTML}{fed700}
\definecolor{tric}{HTML}{21b0fe}

\usepackage[linesnumbered,ruled]{algorithm2e}
\begin{document}

%%
%% The "title" command has an optional parameter,
%% allowing the author to define a "short title" to be used in page headers.
\title{Manifold $k$-NN: Accelerated k-NN Queries for Manifold Point Clouds}

%%
%% The "author" command and its associated commands are used to define
%% the authors and their affiliations.
%% Of note is the shared affiliation of the first two authors, and the
%% "authornote" and "authornotemark" commands
%% used to denote shared contribution to the research.

\author{Pengfei Wang}
\orcid{0000-0002-2079-275X}
\affiliation{%
  \institution{Shandong University}
  \city{Qingdao}
  \country{China}}
\email{pengfei1998@foxmail.com}

\author{Qinghao Guo}
\orcid{0009-0003-7289-3557}
\affiliation{%
  \institution{Shandong University}
  \city{Qingdao}
  \country{China}}
\email{1300165109@qq.com}

\author{Haisen Zhao}
\affiliation{%
  \institution{Shandong University}
  \city{Qingdao}
  \country{China}}
\email{haisenzhao@gmail.com}

\author{Shiqing Xin}
\authornote{Corresponding authors.}
\affiliation{%
  \institution{Shandong University}
  \city{Qingdao}
  \country{China}}
\email{xinshiqing@sdu.edu.cn}

\author{Shuangmin Chen}
\affiliation{%
  \department{School of Information and Technology}
  \institution{Qingdao University of Science and Technology}
  \city{Qingdao}
  \country{China}}
\affiliation{%
  \institution{Shandong Key Laboratory of Deep Sea Equipment Intelligent Networking}
  \city{Qingdao}
  \country{China}}
\email{csmqq@163.com}

\author{Changhe Tu}
\authornotemark[1]
\affiliation{%
  \institution{Shandong University}
  \city{Qingdao}
  \country{China}}
\email{chtu@sdu.edu.cn}

\author{Wenping Wang}
\affiliation{%
  \institution{Texas A\&M University}
  \city{College Station}
  \country{United States of America}}
\email{wenping@tamu.edu}

\renewcommand{\shortauthors}{Wang et al.}

\begin{abstract}
$k$-nearest neighbor ($k$-NN) search is a fundamental primitive in geometry processing and computer graphics. While spatial partitioning structures such as $kd$-trees are standard, they are often manifold-blind, failing to exploit the intrinsic low-dimensional structure of points sampled from 2-manifolds. Recent advances in dynamic programming-based nearest neighbor search (DP-NNS) leverage incrementally constructed Voronoi diagrams to accelerate queries, where each site $p$ maintains a list of \emph{successors} that progressively refine its Voronoi cell. However, DP-NNS is restricted to single nearest neighbor ($k=1$) searches, precluding their adoption in applications that require local neighborhood statistics.

In this paper, we generalize the DP-NNS framework to support arbitrary $k$-NN queries for manifold-aligned data. Our approach is founded on the geometric observation that if $p_i$ is the nearest neighbor of a query $q$ in $P$, then the second nearest neighbor of $q$ must reside either within the prefix set $P_{1:i-1} = \{p_1, \dots, p_{i-1}\}$ or within $p_i$'s successor list. By recursively extending this principle, we introduce \textbf{Manifold $k$-NN}, a recursive algorithmic scheme that significantly outperforms conventional $kd$-trees for manifold-aligned data. Our method achieves a $1\times$--$10\times$ speedup in volume-to-surface query scenarios and inherently supports \emph{dynamic prefix queries}---enabling $k$-NN searches within any subset $P_{1:m}$ ($m \leq n$) with zero overhead.
Furthermore, we extend the framework to support point deletion via local Delaunay updates, providing a complete suite of dynamic operations for point set modification.
Comprehensive experiments on diverse geometric datasets demonstrate the efficiency and broad applicability of our approach for modern graphics pipelines.

Source code is available at \url{https://github.com/sssomeone/manifold-knn}.
\end{abstract}



%%
%% This command processes the author and affiliation and title
%% information and builds the first part of the formatted document.
\maketitle

\section{Introduction}
$k$-nearest neighbor ($k$-NN) search is a fundamental primitive in computer graphics and geometry processing. Its utility spans from classical tasks such as normal estimation and denoising~\cite{10.1145/142920.134011,10.1145/777792.777840,Rui2022RFEPS}, to modern frontiers in implicit surface reconstruction and point cloud registration~\cite{25WenRepair,1175093,121791}. These algorithms rely heavily on querying local neighborhoods to analyze and synthesize underlying geometry. Consequently, the efficiency of $k$-NN search on point-sampled manifolds is often the decisive factor in the overall performance and scalability of modern graphics pipelines.

In many computer graphics applications---such as Moving Least Squares (MLS) projection~\cite{1175093,10.5555/882370.882401} and Signed Distance Field (SDF) evaluation~\cite{10.1145/133994.134011,NeuralPull,ErlerEtAl:Points2Surf:ECCV:2020}---query points are typically distributed throughout a 3D volume, whereas the target data points are constrained to lower-dimensional manifolds. While $k$-NN search traditionally relies on spatial indexing structures like $kd$-trees~\cite{10.1145/361002.361007} or $R$-trees~\cite{10.1145/602259.602266}, their performance often becomes a bottleneck in these \emph{volume-to-surface} query scenarios. This degradation stems from the fact that volumetric partitioning schemes are agnostic to the intrinsic geometric structure of the manifold, leading to inefficient pruning and excessive node traversals (see Figure~\ref{fig:EarthData}).

\begin{figure}[t]
    \centering
    \includegraphics[width=0.99\linewidth]{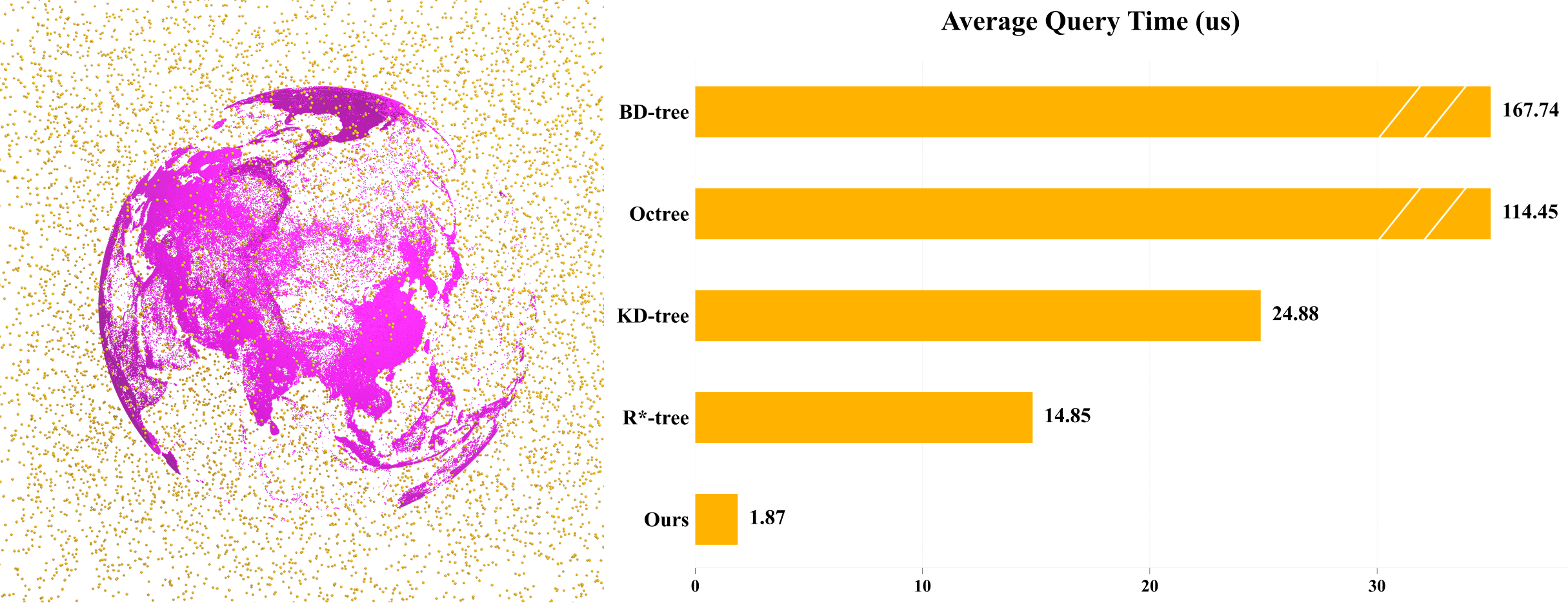}
\caption{Performance comparison in a \emph{volume-to-surface} query scenario. For a dataset of $10^6$ points sampled from a manifold surface (Earth) and $10^6$ query points distributed within its $2\times$ axis-aligned bounding box (AABB), retrieving $k=20$ neighbors exposes the inherent limitations of traditional structures. Standard volumetric indices, such as Octrees and $kd$-trees, suffer from poor pruning and excessive node traversals when constrained by manifold geometry. In contrast, our \textbf{Manifold $k$-NN} achieves an average query time of 1.87\,$\mu$s, demonstrating a substantial speedup by effectively leveraging the intrinsic surface structure.}
    \label{fig:EarthData}
\end{figure}

Recently, \citet{11165079} introduced a dynamic programming-based nearest neighbor search (DP-NNS) framework built upon incremental Voronoi diagrams. Given a point set $P = \{p_1, \dots, p_n\}$ where the indices denote the insertion order, the core mechanism involves maintaining a \emph{successor list} for each site $p_i$. A point $p_j$ ($j > i$) is classified as a successor of $p_i$ if and only if its insertion prunes the Voronoi cell of $p_i$ in the prefix set $P_{1:j}$. While DP-NNS exhibits superior runtime efficiency on manifold-aligned datasets, it is inherently restricted to single nearest neighbor ($k=1$) queries. A na\"{i}ve extension to $k$-NN search—such as performing a breadth-first traversal over the dual Delaunay triangulation starting from the 1-NN—would necessitate redundant geometric predicate evaluations, particularly as $k$ scales.

In this paper, we generalize the DP-NNS framework to support efficient and dynamic $k$-NN queries. Our key geometric observation is that 
the second nearest neighbor of query point $q$ must be either the nearest
neighbor within the prefix $P_{1:i-1} = \{p_1, \dots, p_{i-1}\}$ or one of the
points in $p_i$'s successor list.
By recursively applying this principle, we introduce \textbf{Manifold $k$-NN}, an algorithmic scheme tailored for manifold-aligned data. Our method achieves $1\times$--$10\times$ speedup over $kd$-trees for volume-to-surface queries while remaining competitive for uniform volumetric distributions.
Although Delaunay construction necessitates a more intensive preprocessing phase, the resulting gains in query efficiency represent a worthwhile trade-off in manifold-constrained environments.
Our formulation also naturally supports \emph{dynamic prefix queries}---enabling $k$-NN searches within any subset $P_{1:m}$ ($m \leq n$) without re-computation. This capability is particularly relevant for progressive level-of-detail (LOD) analysis and temporal queries in streaming geometry. Beyond query efficiency, we extend our framework to support site deletion via local Delaunay updates, providing a complete suite of operations for dynamic neighborhood maintenance.

Our main contributions are:
\begin{itemize}[leftmargin=*,nosep]
    \item We generalize the DP-NNS framework to support $k$-NN queries, achieving significant speedups on manifold-sampled geometry compared to state-of-the-art $kd$-trees.
    \item We introduce a robust site deletion operator based on local Delaunay updates, completing the suite of operations required for fully dynamic $k$-NN maintenance.
    \item We demonstrate that our method inherently supports prefix-subset queries with zero overhead, enabling instantaneous access to multiresolution or historical geometric states.
\end{itemize}

\section{Related Work}

\subsection{$k$-Nearest Neighbor Search}
$k$-nearest neighbor ($k$-NN) search is a fundamental operation in geometry processing, serving as a critical building block for tasks ranging from surface reconstruction to point cloud analysis. Traditional approaches primarily rely on hierarchical spatial indexing structures that partition either the ambient space or the data itself to accelerate proximity queries.

The \textit{$kd$-tree}~\cite{10.1145/361002.361007} remains the de facto standard for exact $k$-NN search in low-dimensional static datasets.
By recursively bisecting the point set along alternating coordinate axes, it achieves $O(k \log n)$ query efficiency on averag using priority-queue-based backtracking and geometric pruning~\cite{vermeulen2017comparative}. However, while its partitioning adapts to data density, the axis-aligned splitting strategy is often ``manifold-blind''---it fails to respect the intrinsic geometry of point-sampled surfaces, leading to inefficient pruning and excessive node traversals in volume-to-surface query scenarios. Similarly, the \textit{Octree} decomposes 3D space into regular octants, offering structural simplicity for radius searches and ray casting~\cite{10.1007/s10514-012-9321-0}. Despite adaptive subdivision, fixed axis-aligned partitioning fails to align with manifold geometry, causing significant performance degradation when query points are distant from the surface.

The \textit{$R$-tree}~\cite{10.1145/602259.602266} and its optimized variant, the \textit{$R^*$-tree}~\cite{rstarTREE}, organize data using a hierarchy of Minimum Bounding Boxes (MBBs). Unlike space-partitioning methods that maintain disjoint regions, $R$-trees allow MBBs to overlap, providing flexibility in spatial organization. However, their query efficiency is fundamentally bounded by the tightness of the MBB approximations; when representing thin manifolds embedded in volumetric space, bounding boxes exhibit significant overlap and poor geometric fit, undermining the pruning effectiveness.

Recently, \citet{11165079} proposed a dynamic programming approach (DP-NNS) based on incremental Voronoi diagrams for efficient nearest neighbor search, demonstrating superior performance across diverse point distributions. However, their method is strictly restricted to single nearest neighbor ($k=1$) queries. Generalizing this framework to $k$-NN searches---essential for computing local surface statistics---remains a non-trivial challenge that we address in this work.

\subsection{Voronoi Diagram and Delaunay Triangulation}
\label{sec:Voronoi_and_delaunay}

Given a set of generator sites $P = \{p_i\}_{i=1}^n$ in a domain $\Omega$ with metric $d$, the Voronoi diagram~\cite{Voronoi+1908+97+102} partitions $\Omega$ into cells, where each cell $V(p_i)$ is defined as the region dominated by site $p_i$: 
\begin{equation}
V(p_i) = \{ x \in \Omega \mid d(x, p_i) \leq d(x, p_j), \forall j \neq i \}. 
\end{equation}

This proximity-based decomposition naturally encodes nearest-neighbor relationships. Let $V(p_i)$ denote the Voronoi cell associated with site $p_i \in P$. A query point $q$ resides within $V(p_i)$ if and only if $p_i$ is the nearest neighbor of $q$ in $P$. This property extends to $k$ neighbors through a fundamental geometric principle~\cite{10.5555/1316689.1316762}: for any query location $q$, if $\{p_{i_1}, \dots, p_{i_k}\}$ are its $k$ nearest sites (ordered by ascending distance to $q$), then these sites form a connected subgraph in the Voronoi adjacency graph (i.e., the dual Delaunay triangulation). While this connectivity enables $k$-NN retrieval via local graph traversals starting from the 1-NN, such methods typically entail redundant and computationally expensive geometric predicate evaluations, particularly as $k$ scales or the point distribution becomes complex.

The Delaunay triangulation, the geometric dual of the Voronoi diagram, explicitly encodes these adjacency relationships through edge connectivity. Incremental construction algorithms, such as the Bowyer-Watson method~\cite{10.1093/comjnl/24.2.162}, build Delaunay triangulations by inserting points sequentially and updating affected simplices locally, achieving $O(n \log n)$ average complexity in 3D~\cite{10.1145/777792.777823}. This framework inherently supports dynamic operations: point insertion adds new vertices and updates local connectivity, while point deletion (via local retriangulation) maintains structural validity. Our work leverages this dynamic nature to provide a complete suite of operators for manifold-aware $k$-NN maintenance.

\section{Preliminaries}
\label{sec:preliminaries}

Our work generalizes the DP-NNS framework introduced by~\citet{11165079}, which utilizes a dynamic programming approach for nearest neighbor search. This framework is particularly potent for \emph{point-sampled manifolds}---points sampled from surfaces embedded in $\mathbb{R}^3$---as it effectively exploits their intrinsic low-dimensional structure during the search process.

Given a point set $P = \{p_i\}_{i=1}^n$ ordered by insertion time (birth-time) and a query point $q \in \mathbb{R}^3$, let $\Phi_{1:m}(q)$ denote the nearest neighbor to $q$ within the prefix set $P_{1:m} = \{p_i\}_{i=1}^m$, for $1 \leq m \leq n$. The retrieval of this neighbor follows a recursive sequence:
\begin{equation}
\Phi_{1:m}(q) = \begin{cases}
    p_1, & m = 1, \\
    \Phi_{1:m-1}(q), & m > 1 \text{ and } \lVert q - \Phi_{1:m-1}(q) \rVert \leq \lVert q - p_m \rVert, \\
    p_m, & m > 1 \text{ and } \lVert q - \Phi_{1:m-1}(q) \rVert > \lVert q - p_m \rVert.
\end{cases}
\end{equation}
The global nearest neighbor is thus $\Phi_{1:n}(q)$. 

Let $\mathcal{V}_{1:m}$ denote the Voronoi diagram constructed from $P_{1:m}$. By definition, $q$ resides within the Voronoi cell of its current nearest neighbor: $q \in \text{Cell}(\Phi_{1:m}(q); \mathcal{V}_{1:m})$.
A key geometric insight from~\cite{11165079} is that $\Phi_{1:m}(q)$ only deviates from $\Phi_{1:m-1}(q)$ if the insertion of $p_m$ prunes (shrinks) the Voronoi cell of the previous nearest neighbor (see top part of Figure~\ref{fig:voronoi_insertion}). 

This relationship is formalized as:
\begin{equation}
p_m \not\hookrightarrow \text{Cell}(\Phi_{1:m-1}(q); \mathcal{V}_{1:m-1}) \implies \Phi_{1:m}(q) = \Phi_{1:m-1}(q),
\end{equation}
where $\hookrightarrow$ denotes that the insertion of a site prunes the Voronoi cell of an existing site. By extension, for a sequence of subsequent insertions:
\begin{align}
    p_{m+j} \not\hookrightarrow \text{Cell}(\Phi_{1:m}(q); \mathcal{V}_{1:m+j-1}), \quad \forall 1 \leq j \leq k \nonumber \\
    \implies \Phi_{1:m+k}(q) = \Phi_{1:m+k-1}(q) = \dots = \Phi_{1:m}(q).
\end{align}

To operationalize this, a \textbf{successor list} $L_i$ is maintained for each site $p_i$. When a new site $p_m$ prunes the cells of earlier sites (ancestors) during construction, $p_m$ is appended to their respective successor lists. This structure effectively encodes the evolution of the Voronoi diagram (see Figure~\ref{fig:voronoi_insertion}).

\begin{figure}[t]
    \centering
    \includegraphics[width=.99\linewidth]{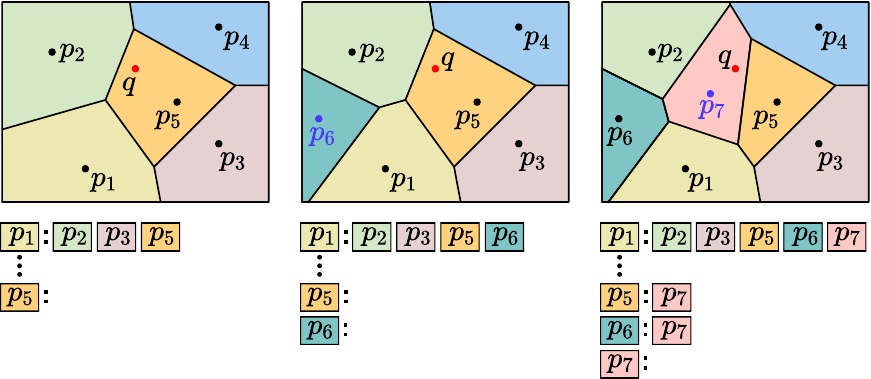}
    \caption{Core observation: a newly inserted point can only change the nearest neighbor of query $q$ if it affects the cell containing $q$. \textbf{Top Left:} Voronoi diagram $\mathcal{V}_5$ from the first 5 points. Query point $q$ (red) lies in the yellow cell with generator $p_5$ ($\Phi_5(q) = p_5$). \textbf{Top Middle:} Inserting $p_6$ creates a cyan cell without affecting the yellow cell, thus $\Phi_6(q) = \Phi_5(q)$. \textbf{Top Right:} Inserting $p_7$ creates a pink cell that shrinks the yellow cell, resulting in $\Phi_7(q) \neq \Phi_6(q)$. Based on this observation, the DP-NNS framework maintains a Query List $L_i$ for each point $p_i$. When inserting $p_j$, the algorithm identifies all Delaunay neighbors of $p_j$ and appends $p_j$ to each neighbor's Query List, encoding which subsequent points have affected each cell. \textbf{Bottom:} Query List states at these insertion stages.}
    \label{fig:voronoi_insertion}
\end{figure}

During query execution, the algorithm traverses these lists starting from $p_1$. If a site $p_j$ in the current list $L_{\text{curr}}$ satisfies $\lVert q - p_j \rVert < \lVert q - p_{\text{curr}} \rVert$, the candidate is updated to $p_j$, and the search jumps to explore $L_j$. \textbf{Essentially, the search path mirrors the sequence of sites that would have occupied $q$'s location throughout the construction of the Voronoi diagram. }

Despite its efficiency for $k=1$ queries, this framework faces two significant challenges:
\begin{enumerate}[leftmargin=*,nosep]
    \item \textbf{Generalization to $k$-NN}: The current logic is restricted to single nearest neighbors, precluding its use in geometry processing tasks that require local neighborhood statistics (e.g., normal estimation, MLS surface fitting).
    \item \textbf{Dynamic Maintenance}: A robust site deletion operator and a mechanism to maintain successor lists during point removal are missing. A complete suite of dynamic operators (insertion and deletion) is essential for modern, evolving graphics datasets.
\end{enumerate}

\begin{figure}[htbp]
	\centering
\includegraphics
[width=1.0\linewidth]{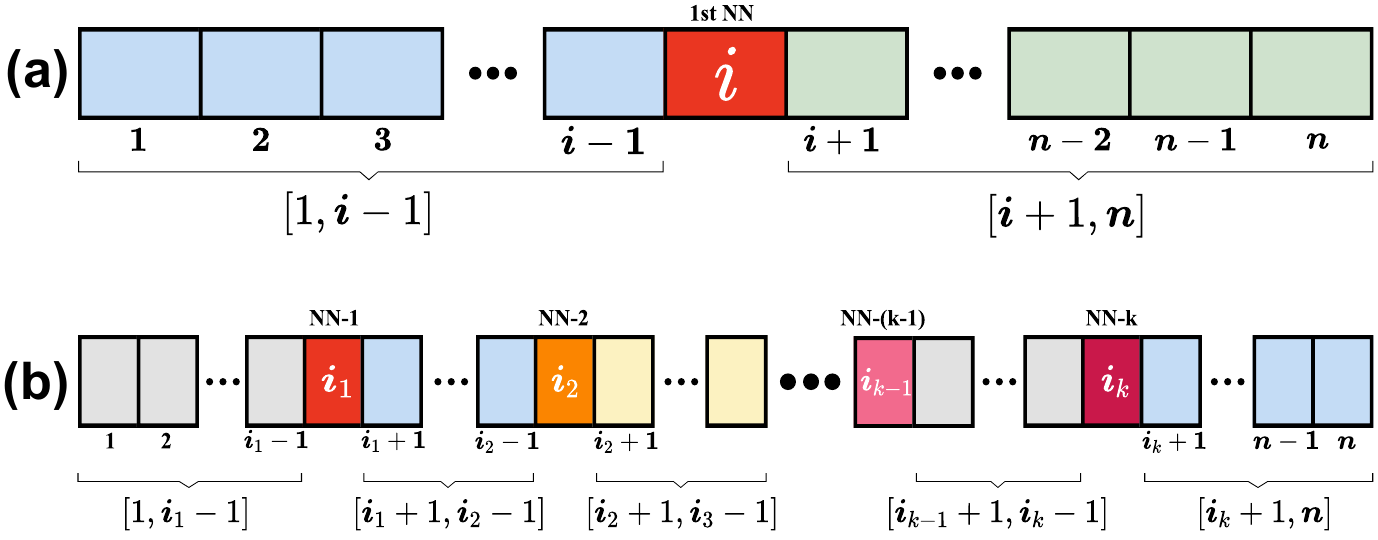}
\caption{Illustration of the recursive search space partitioning for $k$-NN queries. 
\textbf{(a)} Given that the first nearest neighbor is $p_i$, the second nearest neighbor must reside either within the prefix set $P_{1:i-1}$ or within the successor list of $p_i$ (a specific subset of $\{p_{i+1}, \dots, p_n\}$).
\textbf{(b)} For $k$ nearest neighbors identified at indices $\{i_1, i_2, \dots, i_k\}$, the insertion history $[1, n]$ is effectively partitioned into $k+1$ disjoint search intervals. The $(k+1)$-th nearest neighbor is guaranteed to reside in one of these intervals.}

\label{fig:k_nearest}
\end{figure}

\section{Method}
\subsection{State Transition for $k$-Nearest Neighbor Search}
The fundamental insight of our approach is that $k$-NN search on birth-time-ordered point sets can be formulated as a recursive partitioning of the insertion history. 

Formally, let $\Phi_{s:e}(q)$ denote the nearest neighbor to a query point $q$ within the index range $[s, e]$. The global nearest neighbor is thus $p_i = \Phi_{1:n}(q)$. 
As illustrated in Figure~\ref{fig:k_nearest}(a), by definition, the second-nearest neighbor must reside either in the prefix interval $[1, i-1]$ or the suffix interval $[i+1, n]$ (with the convention that $[a, b] = \varnothing$ if $a > b$). We analyze them separately:

\textbf{(1)} If the second-nearest neighbor resides within the interval $[1, i-1]$, it is identical to the nearest neighbor of the prefix set, denoted as $p_j = \Phi_{1:i-1}(q)$. This implies that $p_j$ was the closest site to $q$ at the moment immediately preceding the insertion of $p_i$. Specifically, $p_j$ can be efficiently retrieved by executing a standard 1-NN query (as described in Section~\ref{sec:preliminaries}) restricted to the prefix $P_{1:i-1}$. Note that for this site, $p_j = \Phi_{1:j}(q)$ inherently holds.

\textbf{(2)} If the second-nearest neighbor resides within the suffix $[i+1, n]$, we leverage the Voronoi adjacency properties discussed in Section~\ref{sec:Voronoi_and_delaunay}. Let $p_j$ be the second-nearest neighbor where $j > i$. Since $p_i$ and $p_j$ are the first and second nearest neighbors in the prefix $P_{1:j}$, they must be adjacent in the corresponding Voronoi diagram $\mathcal{V}_{1:j}$. 
Consequently, the insertion of $p_j$ must prune the Voronoi cell of $p_i$.
By the definition of successor lists in Section~\ref{sec:preliminaries}, this guarantees that $p_j$ is a member of $p_i$'s successor list $L_i$.

Consequently, once the first nearest neighbor is identified, the search for the second nearest neighbor proceeds recursively within restricted subsets of the insertion history, requiring no precomputation beyond the initial successor lists. This recursive structure allows us to transform the $k$-NN problem into a sequence of coordinated 1-NN queries over dynamically partitioned index intervals (see Figure~\ref{fig:k_nearest}(b)), as formalized in the following theorem.

\begin{theorem}
\label{thm:knn}
Given a birth-time-ordered point set $P = \{p_i\}_{i=1}^n$, let $\mathcal{N}_k(q) = \{p_{i_1}, p_{i_2}, \dots, p_{i_k}\}$ denote the set of $k$ nearest neighbors of a query $q$ in $P$, 
where indices are sorted by insertion time (birth-time order) such that $i_1 < i_2 < \dots < i_k$.
Then the $(k+1)$-th nearest neighbor of $q$ must either reside within the prefix set $P_{1:i_1-1}$ or belong to the successor list $L_{i_j}$ for some $1 \leq j \leq k$.
\end{theorem}
\begin{proof}
The indices of the $k$ nearest neighbors $\{i_1, i_2, \dots, i_k\}$ partition the insertion history $\{1, 2, \dots, n\}$ into $k+1$ disjoint intervals (some of which may be empty), as illustrated in Figure~\ref{fig:k_nearest}(b): 
\begin{equation}
[1, i_1-1], \quad [i_1+1, i_2-1], \quad \dots, \quad [i_{k-1}+1, i_k-1], \quad [i_k+1, n]. 
\end{equation} 
Let $\ell$ denote the index of the $(k+1)$-th nearest neighbor of $q$ in $P$. By definition, $\ell$ must reside within exactly one of these intervals. We analyze the possibilities via case analysis:

\textbf{Case 1.} If $i_1 > 1$ and $\ell \in [1, i_1-1]$, then by the definition of nearest neighbors, $p_\ell$ must be the nearest neighbor to $q$ within the prefix set $P_{1:i_1-1}$. Formally, $p_\ell = \Phi_{1:i_1-1}(q)$. Furthermore, since $\ell < i_1$ and $i_1$ is the smallest index among the $k$ closer neighbors, it follows that $p_\ell = \Phi_{1:\ell}(q)$, meaning $p_\ell$ was the first nearest neighbor at its time of insertion.

\textbf{Case 2.} If $i_1=1$ or $\ell \notin [1, i_1-1]$, without loss of generality, assume $i_j < \ell < i_{j+1}$ (where $i_{k+1} := n$). In the prefix set $P_{1:\ell}$, the set of $j+1$ nearest neighbors to $q$ consists of $\{p_{i_1}, p_{i_2}, \dots, p_{i_j}, p_\ell\}$, where $p_\ell$ is the $(j+1)$-th nearest. According to the Voronoi adjacency property established in Section~\ref{sec:Voronoi_and_delaunay}, the insertion of $p_\ell$ pruned the Voronoi cell of some $p_{i_m}$ for $m \in \{1, \dots, j\}$. By the definition of successor lists in Section~\ref{sec:preliminaries}, $p_\ell$ must reside within the successor list of one of the $k$ nearest neighbors.

As these cases exhaust all possible intervals for $\ell$, the proof is complete.
\end{proof}

\begin{algorithm} 
\LinesNumbered
\DontPrintSemicolon
\SetKwComment{tcp}{}{}

\KwIn{Point set $P = \{p_1, \ldots, p_n\}$, Successor Table $\{\mathcal{L}_i\}_{i=1}^n$, query point $q$, specified number $k$.
} 
\KwOut{Ordered list $\mathcal{N}$ of at most $k$ transition sites satisfying 
$\Phi_{1:\bullet}(q) = \bullet$, sorted by their Euclidean distance to~$q$.}
{
$\mathcal{N} \gets$ empty list of capacity $k$\;

    $p_{\text{curr}}:=p_1$

    $\mathcal{L}_{\text{curr}}:=\mathcal{L}_{p_1}$

  Insert $p_{\text{curr}}$ into $\mathcal{N}$ 
  
  \ForEach{$p\in\mathcal{L}_{\text{curr}}$}{
    \If{$\|q-p\| < \|q-p_{\text{curr}}\|$}{

          $p_{\text{curr}}:=p$

    $\mathcal{L}_{\text{curr}}:=\mathcal{L}_p$

    Insert $p_{\text{curr}}$ into $\mathcal{N}$ 
    
    }
  }

\KwRet\ $\mathcal{N}$;
}
\caption{Efficient collection of transition sites during 1-NN query.}
\label{alg:TransitionPointsCollectionQuery} 
\end{algorithm}

\subsection{Dynamic Programming Algorithm for $k$-Nearest Neighbor Search}

We formalize the $k$-NN search procedure by synthesizing the geometric cases discussed previously. 

\paragraph{Transition Sites}
According to Theorem~\ref{thm:knn}, if a subsequent nearest neighbor $p_\ell$ (where $\ell$ corresponds to the $2\textsuperscript{nd}, 3\textsuperscript{rd}, \dots, k\textsuperscript{th}$ neighbor) resides within the prefix interval (Case 1), it must satisfy the condition $p_\ell = \Phi_{1:\ell}(q)$. This implies that $p_\ell$ was the first nearest neighbor to $q$ at its specific time of insertion.
Consequently, $p_\ell$ must be one of the sites encountered during the incremental search for the first nearest neighbor.
This observation motivates us to cache these sites, which we term \textbf{Transition Sites} (defined as sites $p_j$ satisfying $p_\bullet = \Phi_{1: \bullet }(q)$), during the initial 1-NN search. By tracking these sites, we bypass the need for an exhaustive search within the prefix interval $[1, i_1-1]$. While the most recently identified Transition Site $p_\ell$ with $\ell < i_1$ is a primary candidate for the next neighbor, a more robust implementation considers all collected Transition Sites as potential candidates. Algorithm~\ref{alg:TransitionPointsCollectionQuery} details the efficient collection of these sites during a standard 1-NN query.

\paragraph{DP-based $k$-NN Search}
Given a fixed value of $k$, Algorithm~\ref{alg:knn} outlines the comprehensive procedure for identifying the full set of neighbors. We maintain a sorted candidate list $\mathcal{N}$ of maximum capacity $k$ to store the best sites found thus far, ordered by their Euclidean distance to $q$. The algorithm utilizes an outer loop to iteratively confirm the $i$-th nearest neighbor for $i = 1, \dots, k$. Upon the confirmation of each $i$-th neighbor, the inner loop traverses its respective successor list. Each successor is evaluated and inserted into $\mathcal{N}$ at its appropriate rank. Any site failing to rank within the top $k$ distances is naturally pruned, ensuring that the search remains focused only on the most promising geometric candidates.

\paragraph{Complexity Analysis}
The preprocessing phase is dominated by incremental Delaunay construction, 
requiring $O(n \log n)$ average time in 3D~\cite{10.1145/777792.777823}, with a 
worst-case complexity of $O(n^2)$ that rarely occurs in practice. For query 
complexity, \citet{11165079} showed that 1-NN search runs in 
$O(\log n)$ expected time. Since the average length of a successor list is 
$O(\log n)$ and our algorithm explores successors for each of the $k$ neighbors, 
the expected complexity of $k$-NN search is $O(k \log n)$.

\paragraph{Prefix-Subset and Multiresolution Queries} 

A distinct advantage of this formulation is that by ignoring candidates with indices exceeding a prescribed threshold $m$, the search is strictly confined to the prefix subset $P_{1:m}$. Consequently, our method inherently supports \emph{prefix-subset queries} with zero overhead. If the birth-time labels are organized according to a spatial hierarchy (e.g., a coarse-to-fine sampling order), our algorithm facilitates $k$-NN queries at any desired geometric resolution without requiring modifications to the underlying data structure.

\begin{algorithm}
\LinesNumbered
\DontPrintSemicolon
\SetKwComment{tcp}{}{}
\KwIn{Point set $P = \{p_1, \ldots, p_n\}$, Successor Table $\{\mathcal{L}_i\}_{i=1}^n$, query point $q$, specified number $k$}
\KwOut{$k$ nearest neighbors of $q$}

Initialize $\mathcal{N}$ by Algorithm~\ref{alg:TransitionPointsCollectionQuery}\;

\tcp*[l]{\textcolor{myorange}{$\triangleright$ The 1st nearest neighbor has been determined}}
\For{$i = 1$ to $k-1$}{
    \tcp*[l]{\textcolor{myorange}{$\triangleright$ Take the $i$-th nearest site that has been determined at this moment}}
    $p \gets \mathcal{N}[i]$\;
    
    \ForEach{$p_\text{successor} \in \mathcal{L}_p$}{
        \If{$p_\text{successor} \notin \mathcal{N}$}{
            Insert $p_\text{successor}$ into $\mathcal{N}$ (ordered by distance to~$q$)\;
        }
    }
}

\KwRet\ $\mathcal{N}$;
\caption{$k$-Nearest Neighbor Query}
\label{alg:knn}
\end{algorithm}

\begin{algorithm}
\LinesNumbered
\DontPrintSemicolon
\SetKwComment{tcp}{}{}
\KwIn{Point set $P = \{p_1, \ldots, p_n\}$, Successor Table $\{L_i\}_{i=1}^n$, point $p_i$ to delete}
\KwOut{Updated Successor Table}

\tcp*[l]{\textcolor{myorange}{$\triangleright$ Neighbors of $p_i$ at insertion; mutual adjacencies unchanged by deletion}}
$N \gets \{p_j \mid p_i \in L_j\}$\;

$\mathcal{D} \gets$ Delaunay triangulation of $N$\;

\ForEach{$p_k \in L_i$}{
    Insert $p_k$ into $\mathcal{D}$\;
    
    $A \gets$ Points in $\mathcal{D}$ adjacent to $p_k$\;
    
   \ForEach{$p_j \in A$}{
    \tcp*[l]{\textcolor{myorange}{$\triangleright$ Skip if already adjacent before deletion}}
    \If{$p_k \notin L_j$}{
    \tcp*[l]{\textcolor{myorange}{$\triangleright$ Verify if adjacency is induced by $p_i$'s removal}}
        $V \gets$ Voronoi vertices shared by $p_j$ and $p_k$\;
        
        \If{$\forall v \in V$ such that $\|v - p_i\| < \|v - p_k\|$}{
            % Append $p_k$ to $L_j$\;
            Insert $p_k$ into $L_j$ maintaining insertion order\;
        }
    }
}
}

\tcp*[l]{\textcolor{myorange}{$\triangleright$ Clean up references to $p_i$}}
\ForEach{$p_j \in N$}{
    Remove $p_i$ from $L_j$\;
}
Delete $L_i$\;

\KwRet Updated $\{L_j, \forall j\neq i\}$\;
\caption{Point Deletion}
\label{alg:deletion}
\end{algorithm}

\subsection{Site Deletion via Local Updates}

While the DP-NNS framework~\cite{11165079} inherently supports point insertion through incremental successor table construction, site deletion poses a significant challenge. Because the successor lists are deeply coupled with the specific construction history, removing a site is non-trivial, as it potentially invalidates the ``birth-time'' dependencies of its descendants.

We address this by exploiting the locality of geometric dependencies. The influence of any site $p_i$ is spatially and topologically confined by its surrounding Delaunay cells. This locality allows us to identify and update only the affected successor lists via local retriangulation, bypassing the need for a costly global reconstruction. Our approach ensures that the state of the successor table following the deletion of $p_i$ is mathematically equivalent to the configuration resulting from the sequential insertion of the point set $P \setminus \{p_i\}$ in its original relative order.

\begin{lemma} \label{lem:affected_lists1} 
When a site $p_i$ is deleted from the point set $P$, any existing successor list entry $p_x \in L_y$ remains valid, provided that $x \neq i$ and $y \neq i$. That is, adjacency relationships in the construction history that do not involve $p_i$ are invariant under its removal. 
\end{lemma}

\begin{proof} 
Consider an entry $p_x \in L_y$ where $x, y \neq i$. By definition, this entry signifies that sites $p_x$ and $p_y$ are adjacent in the Voronoi diagram $\mathcal{V}_{1:x}$ constructed from the prefix set $P_{1:x} = \{p_1, \dots, p_x\}$. We evaluate the impact of deleting $p_i$ based on its insertion order relative to $x$:

\textbf{Case 1: $i > x$.} In this case, $p_i \notin P_{1:x}$. The site $p_i$ was not present in the construction history at the moment the adjacency between $p_x$ and $p_y$ was established. Consequently, its subsequent removal from the global set $P$ cannot retroactively alter the topology of $\mathcal{V}_{1:x}$.

\textbf{Case 2: $i < x$.} Here, $p_i$ is a member of the prefix set $P_{1:x}$. However, a fundamental monotonicity property of Voronoi diagrams states that the removal of a site may cause existing Voronoi cells to expand and potentially form new adjacencies, but it cannot destroy existing adjacencies between the remaining sites. Since $p_x$ and $p_y$ were already adjacent in $\mathcal{V}_{1:x}$ in the presence of $p_i$, they are guaranteed to remain adjacent in the modified diagram $\mathcal{V}_{1:x} \setminus \{p_i\}$.

In both cases, the geometric condition for the entry $p_x \in L_y$ is preserved, confirming its continued validity in the updated successor table.
\end{proof}

Consequently, the reconstruction process following the deletion of $p_i$ reduces to two specific maintenance tasks: 
\begin{enumerate}
    \item \textbf{Pruning:} Removing $p_i$ from all successor lists that contain it, i.e., $\{L_j \mid p_i \in L_j\}$.
    \item \textbf{Redistribution:} Identifying the new adjacencies that emerge to fill the region previously occupied by $p_i$'s Voronoi cell.
\end{enumerate}

\begin{figure}[htbp]
    \centering    
    \includegraphics[width=.99\linewidth]{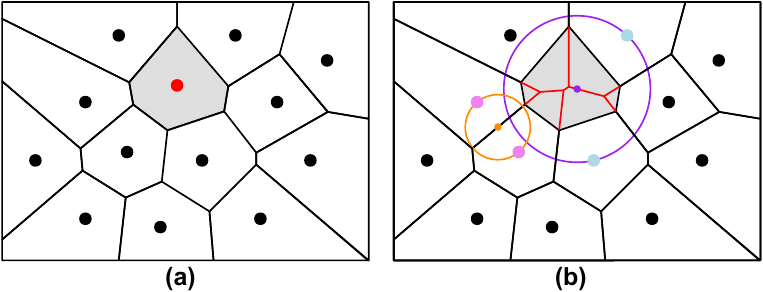}
    \caption{Voronoi adjacency evolution following site deletion. \textbf{(a)} The original diagram highlighting the target site (red) and its corresponding Voronoi cell (gray). \textbf{(b)} The updated diagram after removal, where new Voronoi edges (red lines) emerge to partition the vacated region. The local topology changes are governed by the empty circumcircle property: if two sites originally share a Voronoi edge, their defining circumcircle remains empty after deletion, thus preserving their adjacency. Conversely, the circumcircle associated with previously non-adjacent sites (cyan) becomes empty only upon the removal of the red site, signifying the formation of a new adjacency relationship that fills the vacancy.}
    \label{fig:pointDelete}
\end{figure}

\begin{lemma} \label{lem:new_adjacency} 
Let $\mathcal{V}(P)$ be the Voronoi diagram of the point set $P = \{p_1, \dots, p_m\}$. When a site $p_i$ is deleted from $P$, suppose two sites $p_x, p_y \in P \setminus \{p_i\}$ become adjacent in $\mathcal{V}(P \setminus \{p_i\})$ but were not adjacent in the original diagram $\mathcal{V}(P)$. Then the following properties hold:
\begin{enumerate} 
    \item Every point on the Voronoi edge shared by $p_x$ and $p_y$ in $\mathcal{V}(P \setminus \{p_i\})$ lies within the interior of $\mathrm{Cell}(p_i, \mathcal{V}(P))$.
    \item Both $p_x$ and $p_y$ are Voronoi neighbors of $p_i$ in $\mathcal{V}(P)$. 
\end{enumerate} 
\end{lemma}

\begin{proof} 
\textbf{Part 1.} Let $c$ be any point on the Voronoi edge shared by $p_x$ and $p_y$ in $\mathcal{V}(P \setminus \{p_i\})$, with $r = \|c - p_x\| = \|c - p_y\|$. By the empty sphere property of Voronoi diagrams, the open ball $B(c, r)$ contains no points from the set $P \setminus \{p_i\}$: 
\begin{equation} 
\|c - p_k\| > r, \quad \forall p_k \in P \setminus \{p_i, p_x, p_y\}. 
\end{equation}

We prove $c \in \mathrm{Cell}(p_i, \mathcal{V}(P))$ by contradiction. W.l.o.g, we assume $c \notin \mathrm{Cell}(p_i, \mathcal{V}(P))$. This implies that $p_i$ is not the unique nearest neighbor to $c$ in $P$; thus, $\|c - p_i\| \geq r$. Consequently, the ball $B(c, r)$ would contain no points from the original set $P$ (as illustrated in Figure~\ref{fig:pointDelete}b). The existence of such an empty ball passing through $p_x$ and $p_y$ implies they were already adjacent in $\mathcal{V}(P)$, which contradicts our initial hypothesis. Therefore, we must have $\|c - p_i\| < r$, confirming that $c \in \mathrm{Cell}(p_i, \mathcal{V}(P))$.

\textbf{Part 2.} A fundamental property of Voronoi diagrams states that for any point $q$ strictly inside $\mathrm{Cell}(p_i)$, its second-nearest neighbor in $P$ must be a Voronoi neighbor of $p_i$. As established in Part 1, the point $c$ lies within $\mathrm{Cell}(p_i, \mathcal{V}(P))$. Since $p_x$ and $p_y$ are the nearest neighbors to $c$ in $P \setminus \{p_i\}$, they necessarily constitute the second-nearest neighbors to $c$ in the original set $P$. By the stated property, both $p_x$ and $p_y$ must have been Voronoi neighbors of $p_i$ in $\mathcal{V}(P)$.
\end{proof}

Following the deletion of $p_i$, we update the successor table using a localized reconstruction strategy derived from Lemma~\ref{lem:affected_lists1} and Lemma~\ref{lem:new_adjacency}. While Lemma~\ref{lem:affected_lists1} ensures that adjacencies not involving $p_i$ are invariant, Lemma~\ref{lem:new_adjacency} establishes that new adjacencies can only emerge between the former neighbors of $p_i$.

Our approach focuses on identifying these new adjacencies within $p_i$'s neighborhood. We extract all sites that were ever adjacent to $p_i$ during the incremental construction history: specifically, sites $p_j$ where $p_i \in L_j$ (pre-insertion neighbors) and sites already in $L_i$ (post-insertion neighbors). 
We then perform local incremental Delaunay construction on this restricted set, which is typically very small—around \textbf{25-35 points for deletion operations on point clouds ranging from 200K to 2M vertices}.
The computational cost is primarily governed by the length of the deleted point's successor list ($O(\log n)$ on average) and its neighborhood size at insertion time ($O(1)$ on average), yielding efficient average-case performance.
To ensure correctness, we filter out spurious adjacencies—those that do not result from $p_i$'s removal—by verifying that the shared Voronoi vertices of each candidate pair lie within the original $\mathrm{Cell}(p_i)$. Algorithm~\ref{alg:deletion} details the complete procedure.

\begin{figure}[htbp]
	\centering
\includegraphics
[width=.99\linewidth]{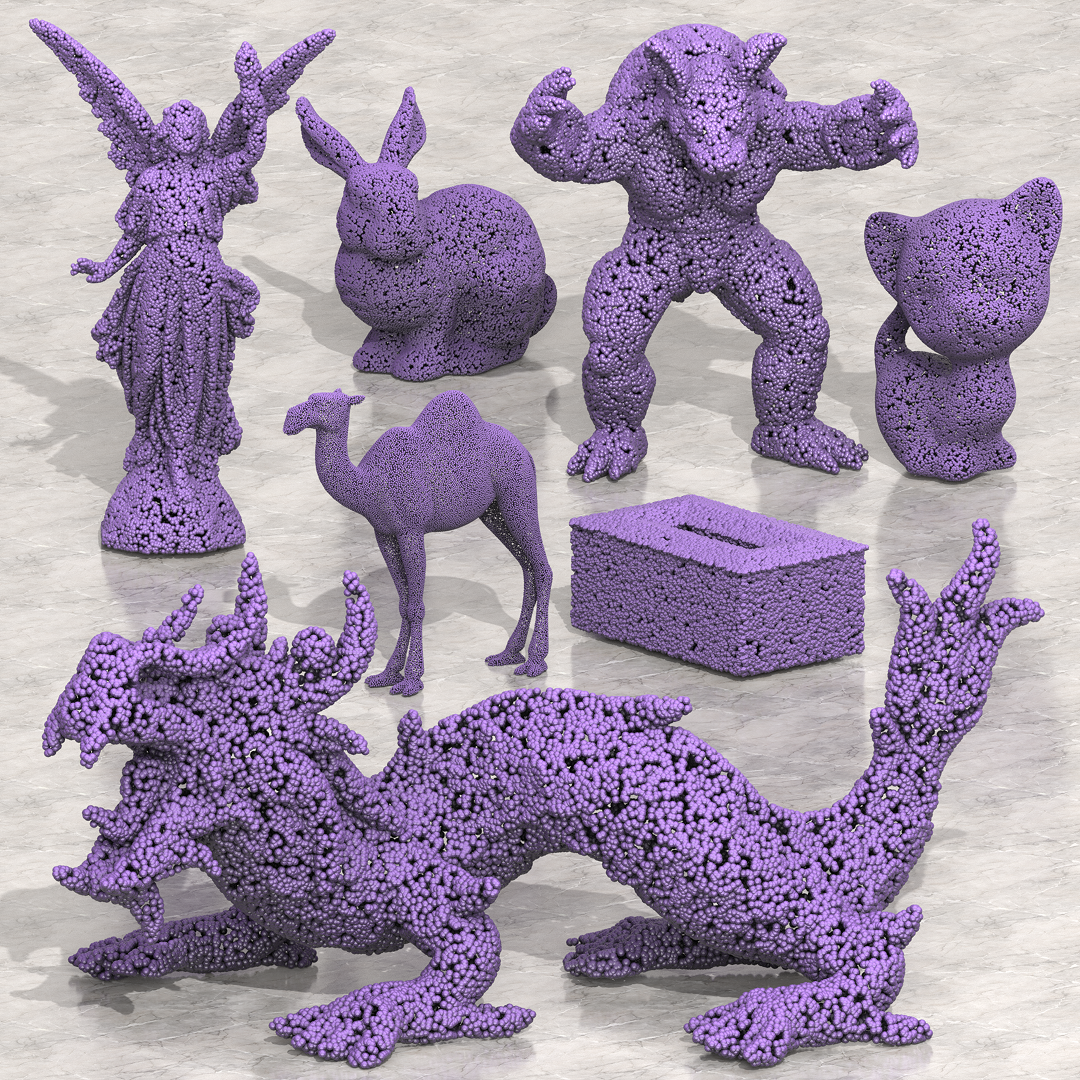}
\caption{Visualization of representative point cloud models used in experiments.}
\label{fig:models}
\end{figure}
\section{Evaluation}

\subsection{Implementation and Platform}
Our C++ implementation was evaluated on a Mac mini with an Apple M4 CPU and 16GB RAM, running macOS. We use CGAL~\cite{cgal:hs-chdt3-24a} for Delaunay triangulation construction. The input point cloud is preprocessed using CGAL's \texttt{spatial\_sort} function, which organizes points into random buckets of increasing sizes with Hilbert sorting applied within each bucket. This strategy accelerates incremental Delaunay construction while maintaining query performance.

In the experimental evaluation, we primarily compare our method against four 
established spatial indexing structures: (1) KD-tree from the nanoflann 
library~\cite{blanco2014nanoflann} with leaf size set to 10; (2) R$^{*}$-tree 
from the Boost library~\cite{BoostLibrary} with maximum leaf capacity of 10; 
(3) Octree from the Point Cloud Library (PCL)~\cite{Rusu_ICRA2011_PCL} with 
resolution set to 3 times the average inter-point distance; and (4) BD-tree 
from the ANN library, configured with a maximum leaf size of $10$ for exact $k$-nearest neighbor queries. Additionally, for particular test cases, we include comparisons with (5) ArborX~\cite{Prokopenko2025TheAL}, a BVH-based geometric search library optimized for large-scale nearest neighbor queries via Morton code-sorted traversal, and (6) a na\"{i}ve Delaunay-traversal baseline (DT-BFS). The latter utilizes the original successor table~\cite{11165079} to identify the $1$-NN, followed by a breadth-first search on the Delaunay adjacency graph to retrieve the remaining $k-1$ neighbors.
Figure~\ref{fig:models} shows several representative point cloud models used in our experiments.

\begin{figure}[htbp]
	\centering
\includegraphics
[width=.99\linewidth]{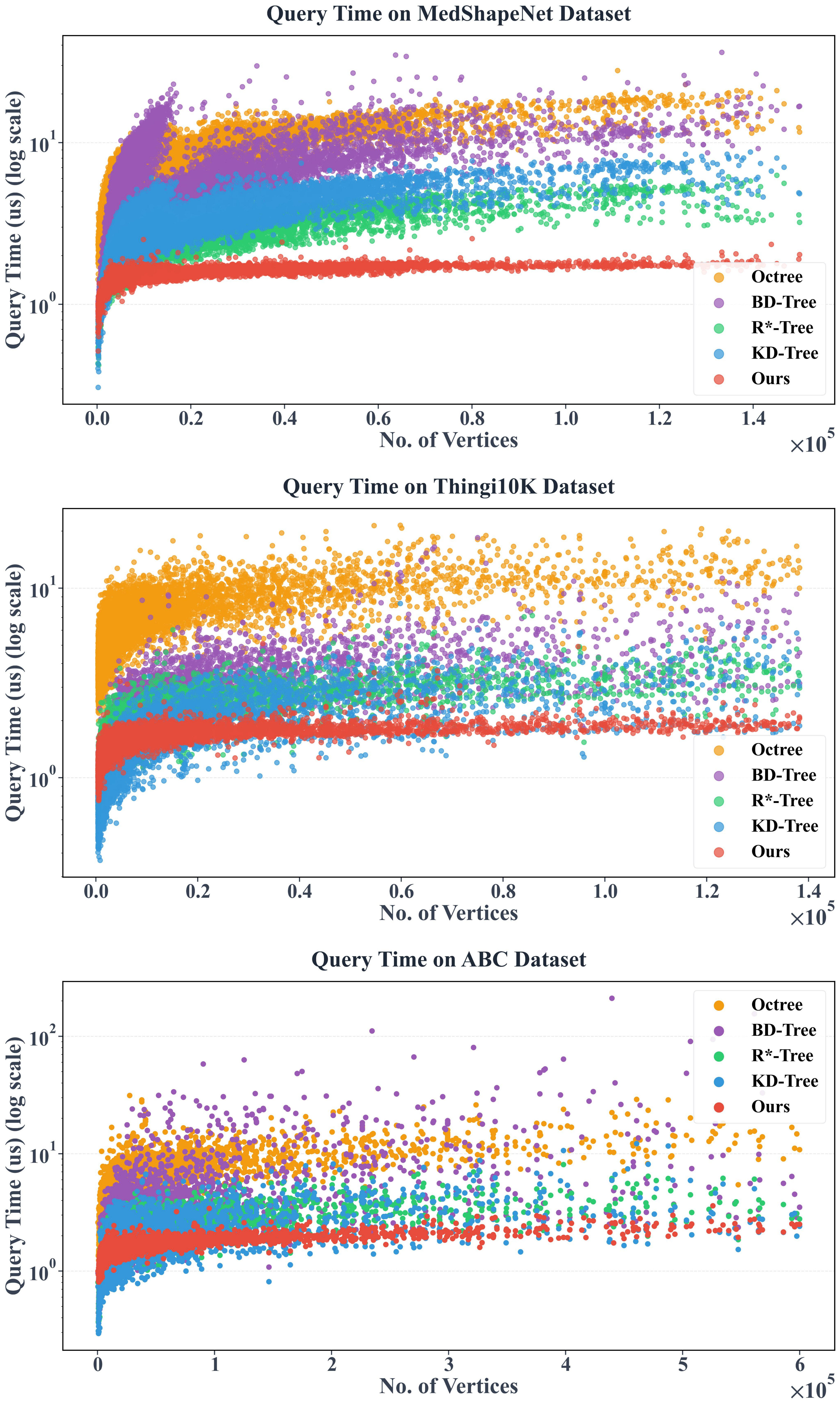}
\caption{Query time comparison across MedShapeNet, Thingi10K, and ABC datasets (from top to bottom).
Performance varies with geometric structure: MedShapeNet 
models exhibit clear low-dimensional manifolds favoring our method, while others 
contain complex geometries with more uniformly distributed points.}
\label{fig:queryrandom}
\end{figure}

\subsection{Static k-NN Performance}

We evaluate $k$-NN query performance in volume-to-surface scenarios, where query points are distributed in volumetric space while target points lie on manifold surfaces—a common setting in geometry processing.
We tested on three large-scale datasets—subsets from ABC~\cite{Koch_2019_CVPR} and MedShapeNet~\cite{li2023medshapenet}, along with Thingi10K~\cite{Thingi10K}—as well as standard 3D models.
For each model, we generated $10^6$ query points uniformly distributed within a $2\times$ axis-aligned bounding box with $k=20$. 
As shown in Figure~\ref{fig:queryrandom} and Table~\ref{table:query_time}, our method achieves $1\times$ to $10\times$ speedup over state-of-the-art spatial indexing structures.
Additionally, since ArborX demonstrates stronger performance on very large point sets, we evaluated our method on a uniform point cloud and the Lucy model, both containing $10$M points, with query points sampled within their $1\times$ and $2\times$ bounding boxes, respectively. Our approach achieves $1\times$ and $3\times$ speedups over ArborX, respectively, while maintaining $>2\times$ and $>10\times$ speedups compared to other state-of-the-art baselines.

\begin{figure*}[htbp]
	\centering
\includegraphics
[width=0.95\linewidth]{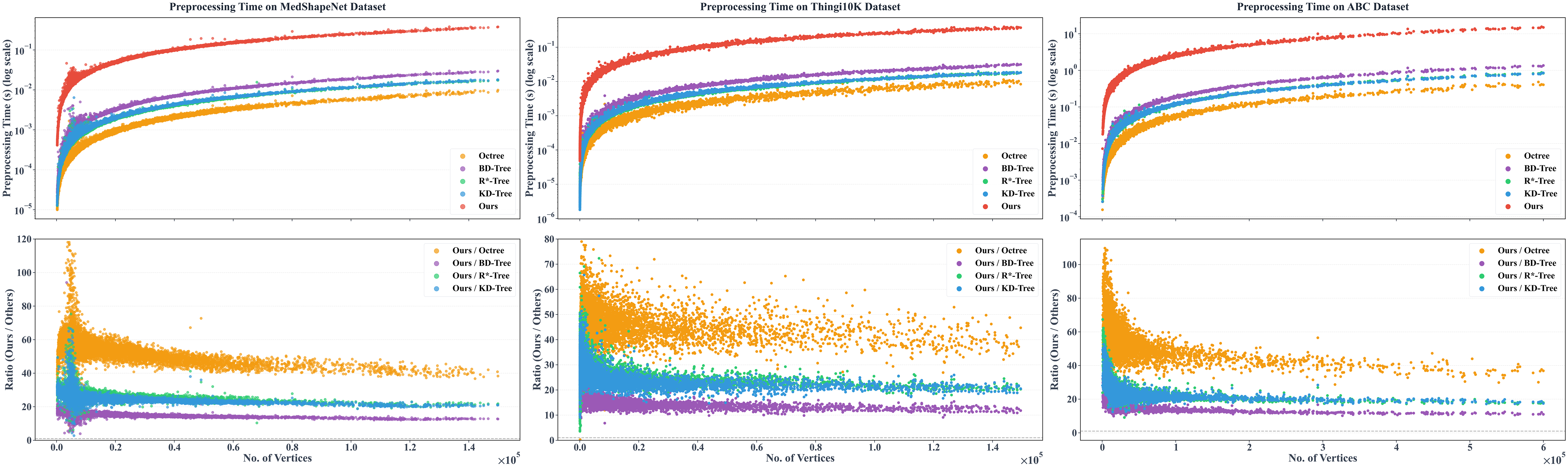}
\caption{Preprocessing time comparison across MedShapeNet, Thingi10K and ABC datasets (from left to right). Dashed line indicates equal performance (ratio =1).}
\label{fig:preprocessingTime}
\end{figure*}

\begin{table}[htbp]
\caption{Average query time ($\mu s$) comparison on different 3D models. \textbf{\underline{Best results}} are shown in bold underline, \underline{second-best} results are underlined.}

\centering\small

\resizebox{0.99\columnwidth}{!}{
  \renewcommand{\arraystretch}{1.25}
  \setlength{\tabcolsep}{7pt}
 
  \rowcolors{2}{white}{rowgray}
  \begin{tabular}{c|ccccccc}
  \hline
  \rowcolor{headerblue} & Camel & Bunny & Dragon & Kitten & Armadillo & Lucy & Sponza \\
  \hline
  Vertices   & 28934 & 72911 & 100313 & 291023 & 726367 & 1018219 & 1313504 \\
  \hline 
  KD-tree    & 2.892 & 5.026 & 3.900 & 9.562 & 11.56 & 10.13 & 4.045 \\
  R$^*$-tree & 2.420 & 3.589 & \ul{2.585} & 6.013 & 6.193 & 4.867 & \ul{3.002} \\
  Octree     & 7.822 & 13.38 & 10.99 & 23.57 & 29.37 & 25.28 & 19.48 \\
  BD-tree    & 4.455 & 7.631 & 8.704 & 18.55 & 118.7 & 38.83 & 7.092 \\
  ArBorX     & \ul{2.319} & 3.744 & 2.987 & 6.608 & 7.675 & 5.772 & 3.103 \\
  DT-BFS     & 2.708 & \ul{2.681} & 2.881 & \ul{3.337} & \ul{3.577} & \ul{4.024} & 5.554 \\
  Ours       & \ul{\textbf{1.704}} & \ul{\textbf{1.718}} & \ul{\textbf{1.797}} & \ul{\textbf{1.789}} & \ul{\textbf{1.972}} & \ul{\textbf{2.264}} & \ul{\textbf{2.843}} \\
  \hline
  \end{tabular}
}
\label{table:query_time}
\end{table}

Figure~\ref{fig:preprocessingTime} compares preprocessing time across different datasets. Our method requires longer construction time than the KD-tree baseline, approximately 20$\times$ slower, primarily due to the overhead of incremental Delaunay triangulation and Successor Table maintenance. 
However, this preprocessing cost is amortized in query-intensive applications, where the structure is built once but queried millions of times, making the trade-off favorable in many practical scenarios.
We highlight two recent representative works in which $k$-NN search over a point cloud is the central computational primitive:
\begin{itemize}

\item \textbf{POCO}~\cite{Boulch_2022_CVPR} (\textit{neural implicit surface reconstruction}) 
reconstructs surfaces from point clouds by decoding occupancy through $k$-NN-based 
feature interpolation, issuing millions of queries over a marching-cubes grid per 
reconstruction. The authors explicitly report $k$-nearest-neighbor computation as 
the dominant runtime bottleneck of their pipeline.

\item \textbf{Point-NeRF}~\cite{Xu2022PointNeRFPN} (\textit{point-based neural rendering}) 
synthesizes novel views by representing a scene as a neural point cloud, aggregating 
local neural features via $k$-NN queries at millions of shading locations along 
camera rays.
\end{itemize}

\begin{table}[htbp]
\caption{Scan termination optimization on progressive scans (100 frames, time in ms). Vertices indicates the total number of points in each model. Our method achieves significant speedup over traditional structures. \textbf{\underline{Best results}} are shown in bold underline, \underline{second-best} results are underlined.}
\centering\small
\resizebox{0.95\columnwidth}{!}{
  \renewcommand{\arraystretch}{1.25}
  \setlength{\tabcolsep}{7pt}
  \rowcolors{2}{white}{rowgray}
  \begin{tabular}{c|cccccc}
  \hline
  \rowcolor{headerblue} & Camel & Bunny & Dragon & Kitten & Armadillo & Lucy \\
  \hline
  Vertices   & 676695 & 1059183 & 459903 & 966442 & 889123 & 459473 \\
  KD-tree    & \ul{381.0} & \ul{545.8} & \ul{340.9} & \ul{520.0} & \ul{525.6} & \ul{336.1} \\
  R$^*$-tree & 402.4 & 585.6 & 364.8 & 542.0 & 559.2 & 360.2 \\
  Ours       & \ul{\textbf{57.83}} & \ul{\textbf{65.99}} & \ul{\textbf{48.03}} & \ul{\textbf{63.86}} & \ul{\textbf{67.15}} & \ul{\textbf{56.43}} \\
  \hline
  \end{tabular}
}
\label{table:scan_optimization}
\end{table}

\subsection{Prefix k-NN Query Performance}

A distinguishing feature of our framework is the ability to query $k$-NN within any prefix $\{p_1, \ldots, p_m\}$ with zero overhead. while traditional structures (KD-tree, R*-tree) require full reconstruction when $m$ changes. We evaluate this capability using synthetic progressive scans generated via Open3D~\cite{Zhou2018}, where a virtual depth sensor (320$\times$240 resolution) traverses a Fibonacci sphere trajectory around standard models (Bunny, Dragon, Armadillo, etc.), producing 100 temporally-ordered frames. This setup simulates progressive acquisition workflows where non-sequential temporal access is beneficial.

\subsubsection{Scan Convergence Analysis}

Determining the minimal scan duration for sufficient geometric coverage is a common calibration task in automated scanning. We apply binary search to identify the minimal frame count $m$ achieving 99\% surface convergence. At each candidate $m$, we sample 2,000 probe points on the ground truth surface and query their $k=20$ nearest neighbors within the first $m$ frames to fit local planes via PCA. A probe converges when its fitting residual falls below 0.1\% of the model diameter. 
This workflow requires evaluating non-sequential temporal states (e.g., testing $m=50$, then $m=75$, then $m=62$), which is prohibitively expensive for traditional structures due to repeated reconstruction. 
Table~\ref{table:scan_optimization} shows our method completes the binary search (6-7 iterations) significantly faster than KD-tree and R*-tree.
This demonstrates the practical value of zero-cost prefix switching in exploratory temporal queries.

\subsubsection{Temporal Defect Localization}

Progressive scans may accumulate artifacts during acquisition. We simulate this by injecting $15\times$ noise into frames 30-40 and examine how local neighborhoods evolve across temporal states. Specifically, we monitor the $k=100$ nearest neighbors of a query point at the suspected defect location while randomly adjusting $m$ to identify when irregularities first emerge.
Our method sustains 5,000--10,000 queries per second during random temporal switching, while KD-tree and R*-tree achieve only 10--20 queries per second due to reconstruction overhead at each $m$ change. This 250--500$\times$ speedup enables interactive exploration of temporal anomalies, transforming what would be a tedious batch analysis into a fluid debugging workflow.

\begin{figure}[htbp]
	\centering
\includegraphics
[width=.99\linewidth]{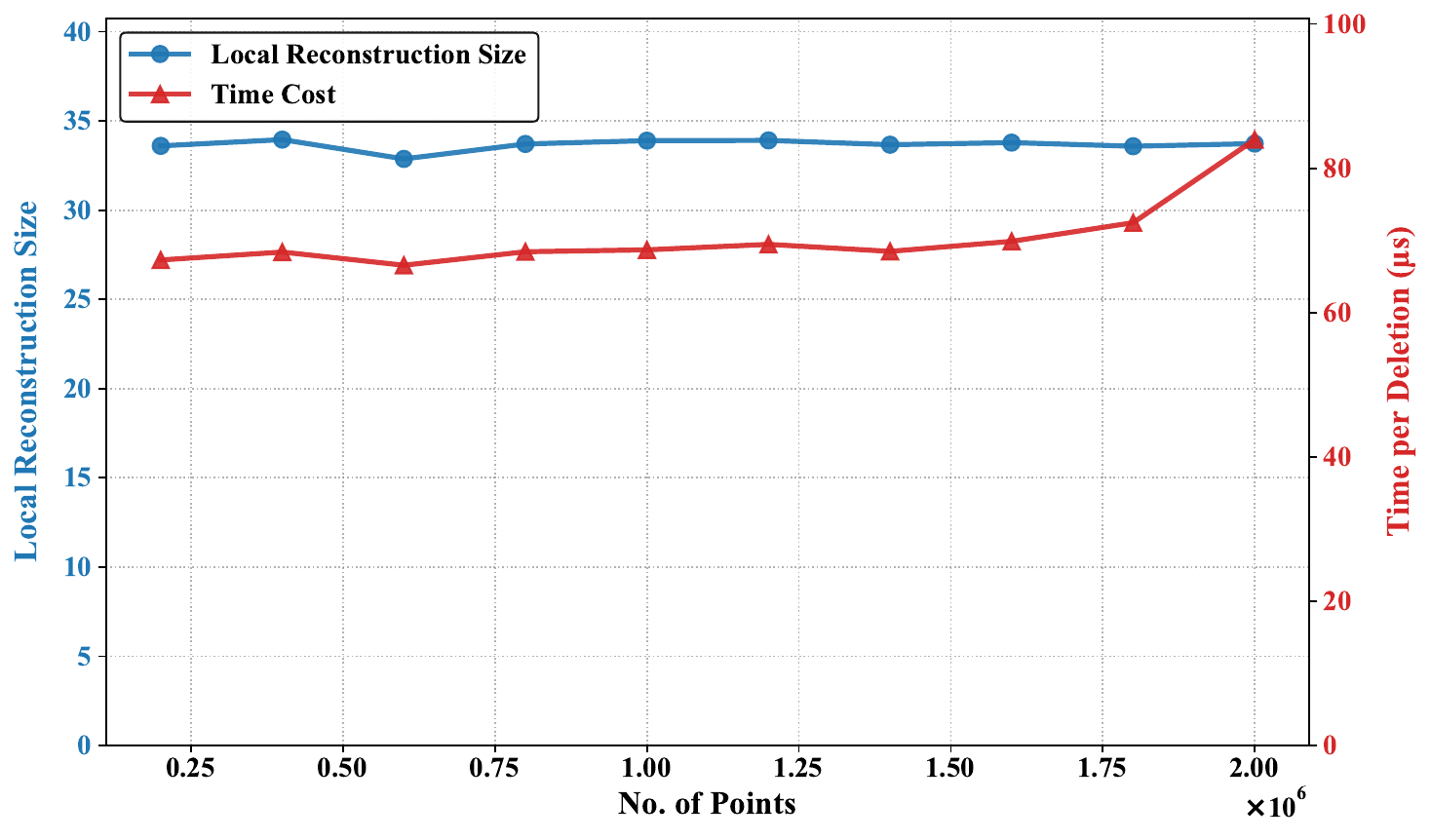}
\caption{Point deletion performance across different point cloud sizes. For each dataset size, we randomly delete 10,000 points and measure the average number of points involved in local Delaunay reconstruction per deletion, and the average deletion time per point. The local reconstruction size remains consistently small (approximately 33 points) regardless of total point cloud size, demonstrating the locality of our deletion strategy.}
\label{fig:deleteData}
\end{figure}

\begin{figure}[htbp]
	\centering
\includegraphics
[width=0.99\linewidth]{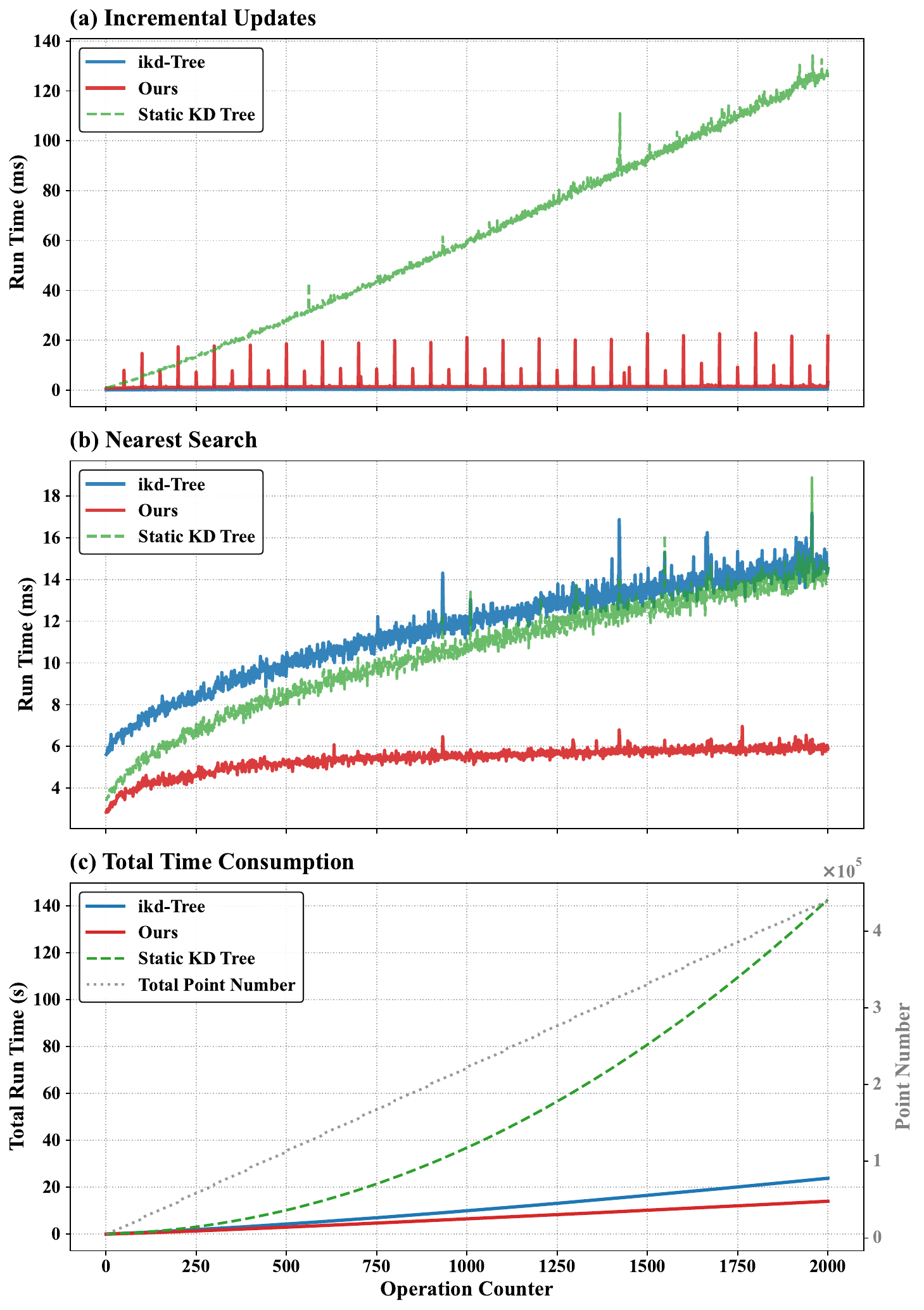}
\caption{Performance comparison under dynamic scenarios. Target points lie on the Dragon model surface while query points are randomly distributed in its $2\times$ axis-aligned bounding box. (a) Time per incremental update operation. (b) Query time per iteration. (c) Cumulative total time versus iteration number. Our method demonstrates superior overall performance compared to iKD-Tree and static KD-tree (nanoflann) for query-intensive workloads.}
\label{fig:dyanmicTest_SamplesOnSurfaceQueryRandom}
\end{figure}

\begin{figure}[htbp]
	\centering
\includegraphics
[width=0.99\linewidth]{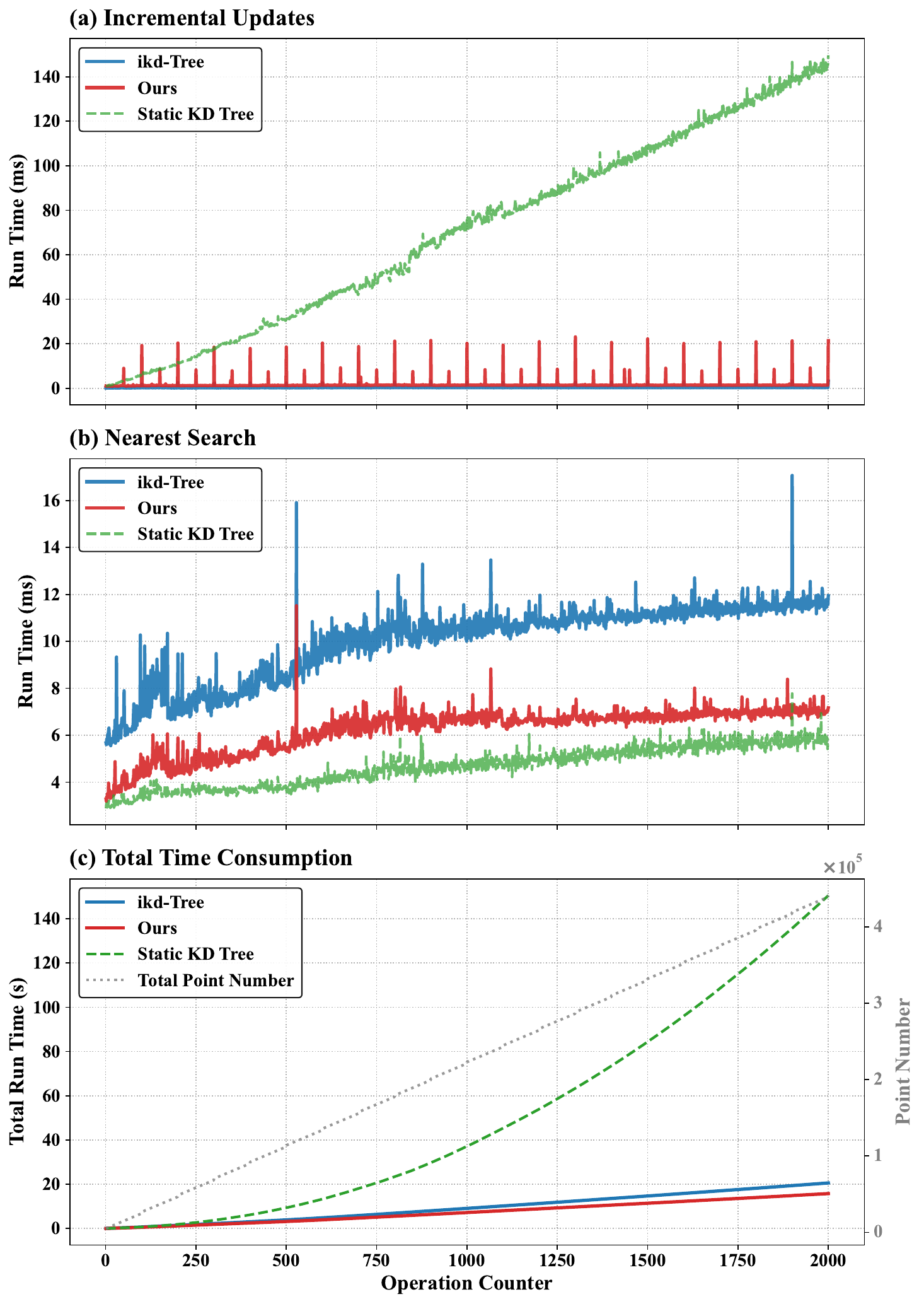}
\caption{Performance comparison under dynamic scenarios. Both target points and query points are uniformly distributed within $[-10, 10]^3$. (a) Time per incremental update operation. (b) Query time per iteration. (c) Cumulative total time versus iteration number. Our method demonstrates superior overall performance compared to iKD-Tree and static KD-tree (nanoflann) for query-intensive workloads.}
\label{fig:dyanmicTest_InputRandomQueryRandom}
\end{figure}

\subsection{Dynamic Point Operations}
\label{sec:DynamicPointOperations}

We extend the DP-NNS framework to support point deletion through local Delaunay reconstruction. When deleting a point, our approach identifies the affected neighborhood and performs incremental Delaunay construction only within this local point set to update the relevant Successor Lists, avoiding global rebuilding. Figure~\ref{fig:deleteData} demonstrates this locality: the affected neighborhood contains only ~$33$ points on average across point cloud sizes from 200K to 2M.
While not optimized for deletion-heavy workloads, our method excels in query-intensive scenarios where superior $k$-NN performance provides overall advantages.

We evaluate this through a dynamic benchmark adapted from ikd-Tree~\cite{xu2022fast}.
Starting with 5,000 randomly generated points, we execute 2,000 test iterations where each iteration performs: (1) insertion of 200 new points, (2) 2,000 query operations with $k=20$, (3) deletion of 100 points every 50 iterations, and (4) batch insertion of 2,000 points every 100 iterations. This configuration reflects workloads where queries and insertions dominate.

We test two scenarios. In the first, both initial and inserted points lie on the Dragon model surface while query points are uniformly distributed within the $2\times$ bounding box (Figure~\ref{fig:dyanmicTest_SamplesOnSurfaceQueryRandom})—representative of volume-to-surface query patterns. In the second, both target and query points are uniformly sampled from $[-10, 10]^3$ (Figure~\ref{fig:dyanmicTest_InputRandomQueryRandom}). Results show that static KD-tree incurs prohibitive reconstruction costs in dynamic settings. Our method achieves superior overall throughput compared to both baselines: the accelerated query performance compensates for slower individual deletion operations, resulting in better end-to-end efficiency in these query-intensive dynamic workloads.

\begin{table}[htbp]
\caption{Average query time ($\mu s$) comparison across different methods for varying $k$ values 
on the Lucy model.}
\centering\small
\resizebox{0.99\columnwidth}{!}{
  \renewcommand{\arraystretch}{1.25}
  \setlength{\tabcolsep}{7pt}
  \rowcolors{2}{white}{rowgray}
  \begin{tabular}{c|ccccccc}
  \hline
  \rowcolor{headerblue} $k$ & KD-tree & R$^*$-tree & Octree & BD-tree & ArBorX & DT-BFS & Ours \\
  \hline
  5  & 7.473 & 4.268 & 21.35 & 37.09 & 3.833 & 1.926 & \textbf{0.929} \\
  10 & 8.573 & 4.501 & 23.87 & 39.05 & 4.849 & 2.605 & \textbf{1.445} \\
  20 & 10.53 & 5.143 & 26.51 & 42.64 & 6.413 & 4.125 & \textbf{2.427} \\
  50 & 18.30 & 6.621 & 34.30 & 55.01 & 10.48 & 8.598 & \textbf{5.368} \\
  \hline
  \end{tabular}
}
\label{table:k_scaling}
\end{table}

\subsection{Scalability and Sensitivity Analysis}

\paragraph{Scalability with k.}
While $k=20$ represents a typical value for many geometry processing applications, we further evaluate how performance scales with neighborhood size by testing $k$ values ranging from $5$ to $50$ on the Lucy model. Because $k$ is generally small in practice, we organize candidates using insertion sort, following the strategy of nanoflann~\cite{blanco2014nanoflann}. This approach is highly efficient for small $k$ but may introduce overhead as $k$ increases. Table~\ref{table:k_scaling} compares our query costs against state-of-the-art methods, demonstrating competitive performance across various values of $k$.

\begin{figure}[htbp]
	\centering
\includegraphics
[width=.99\linewidth]{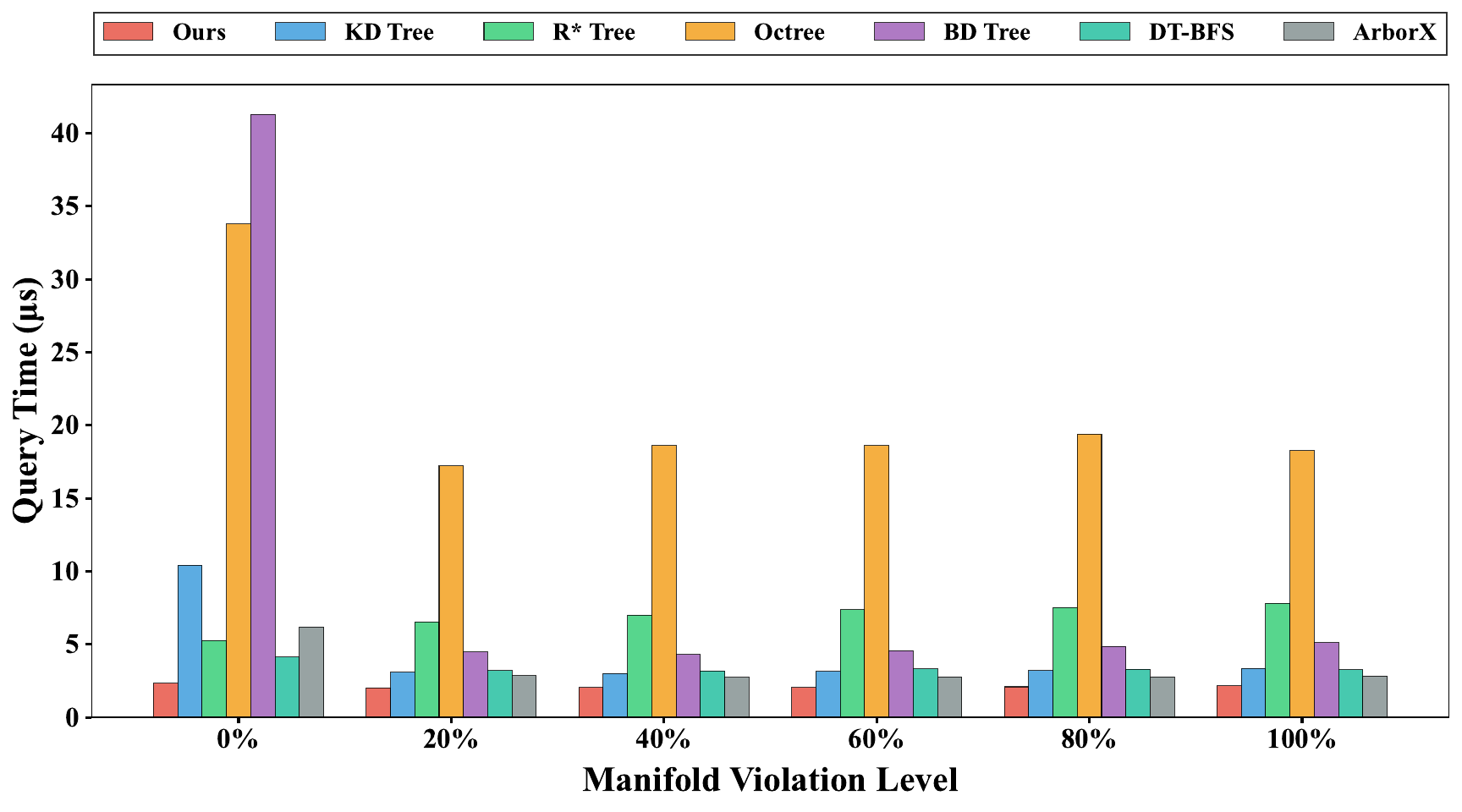}
\caption{Query Time Performance Comparison Across Manifold Violation Levels.}
\label{fig:manifold_degradation}
\end{figure}

\begin{figure}[htbp]
	\centering
\includegraphics
[width=.99\linewidth]{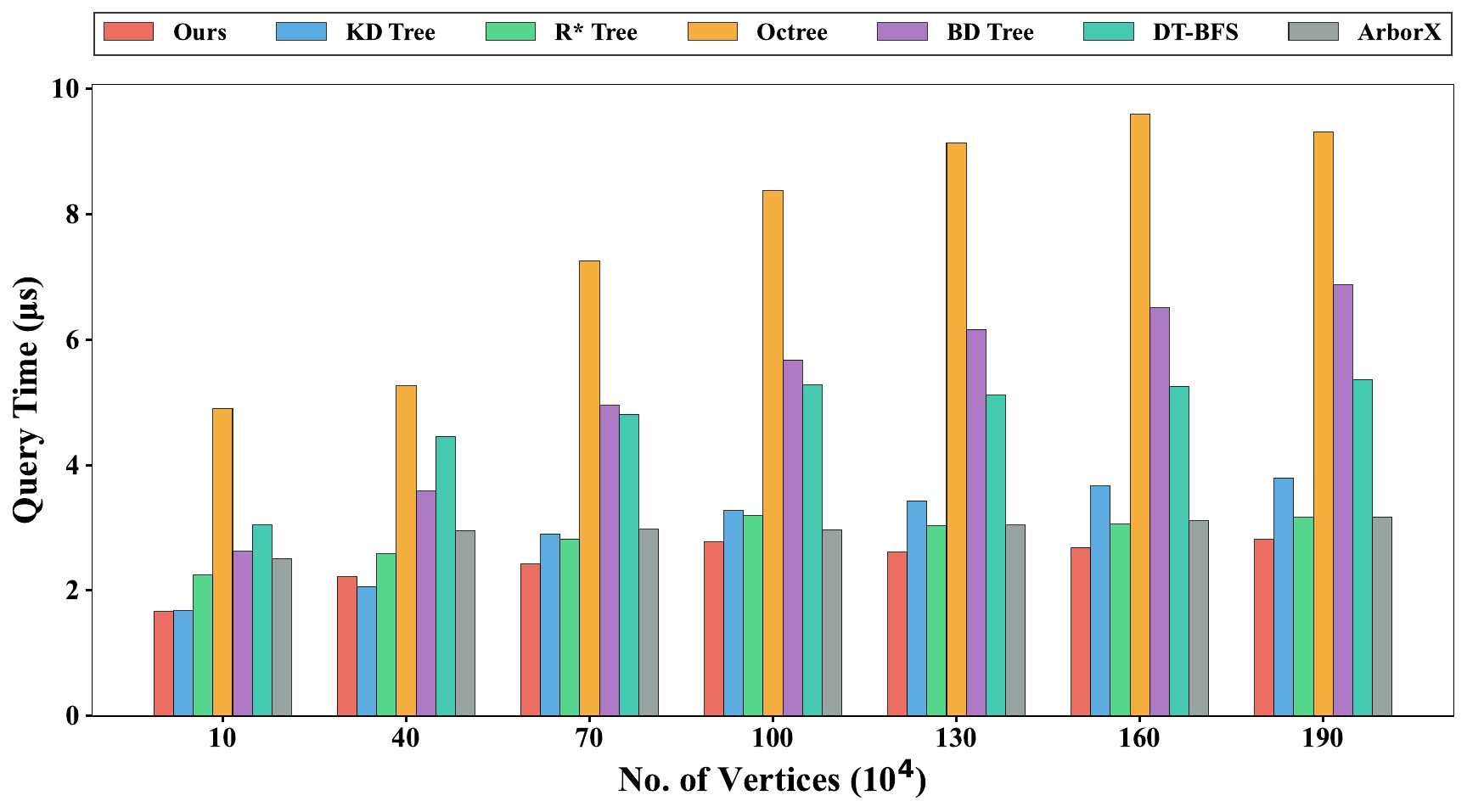}
\caption{Query time comparison on \textbf{uniform point clouds} with varying numbers of target points. All query points were generated within the \textbf{exact same spatial bounds} as the target points. Our method achieves query performance comparable to R$^{*}$-tree, KD-tree and ArborX while outperforming other baseline methods.}
\label{fig:randomPointrandomQuery}
\end{figure}

\paragraph{Sensitivity to Manifold Structure.}
To understand how performance depends on this geometric property, we evaluate query 
efficiency as the manifold structure is progressively degraded. We perturb the Lucy 
model by randomly displacing each point by varying fractions of the bounding box 
diagonal. Figure~\ref{fig:manifold_degradation} compares query time against 
state-of-the-art methods across different displacement ratios. While our advantage 
decreases as the manifold structure weakens, the method maintains competitive performance 
even under substantial geometric perturbation.
Nevertheless, even on fully uniform point clouds—where both query and target points are uniformly sampled within the same cubic volume—our method maintains competitive performance comparable to R$^{*}$-tree, KD-tree and ArborX while significantly outperforming other baselines (Figure~\ref{fig:randomPointrandomQuery}). This demonstrates that our approach remains effective beyond ideal manifold settings.

\begin{figure}[htbp]
	\centering
\includegraphics
[width=.99\linewidth]{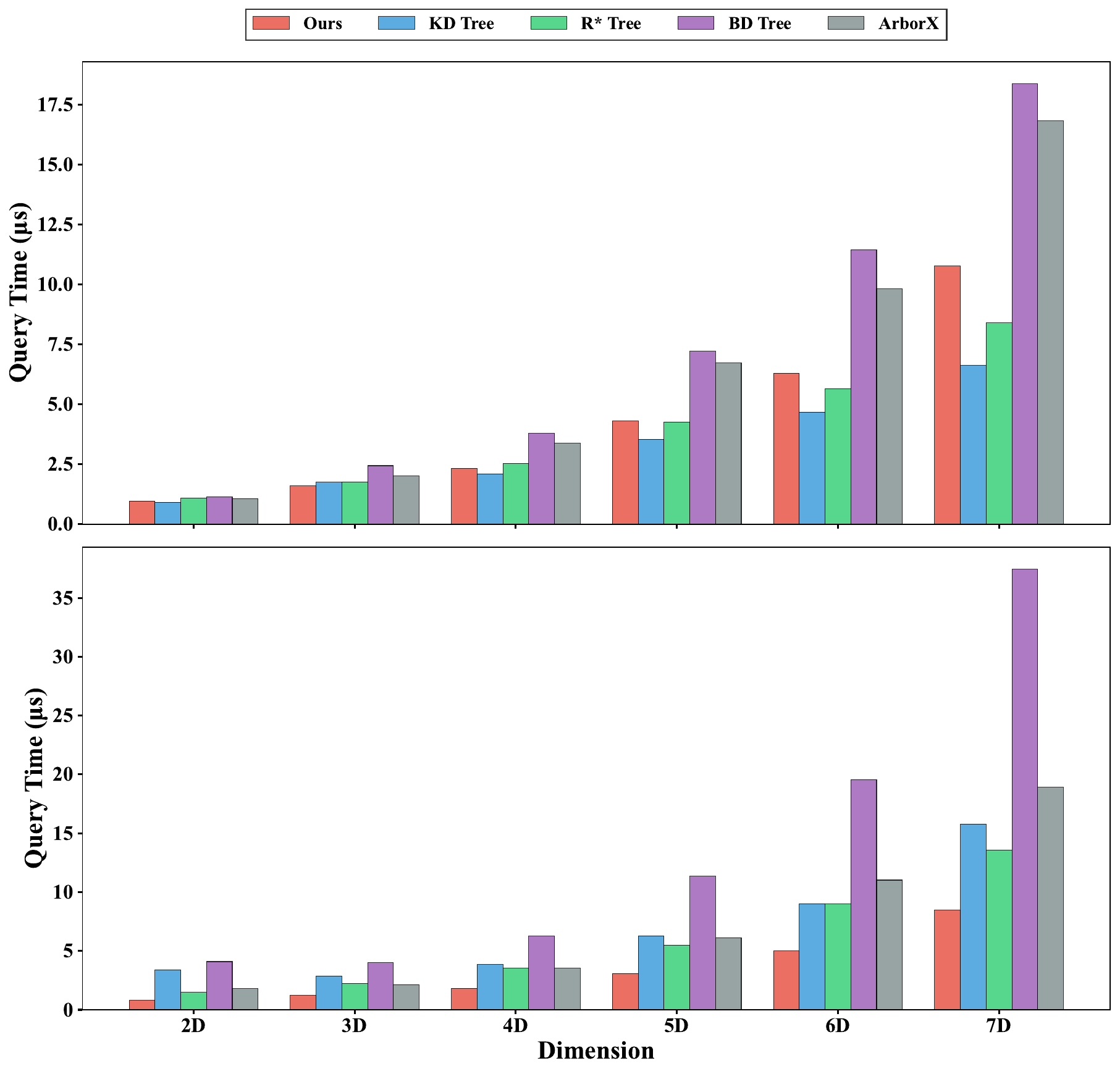}
\caption{Query time comparison across dimensions for uniform distribution (top) 
and manifold distribution on hypersphere (bottom).}
\label{fig:dimension}
\end{figure}

\paragraph{Performance Across Dimensions}
While our method theoretically supports arbitrary dimensions, its primary application domain is 2D/3D due to the exponential complexity of high-dimensional Delaunay construction. Nevertheless, understanding how query performance scales with dimensionality remains instructive. We evaluate two scenarios using $10$K target points and $1$M query points: (1) both uniformly distributed in $[-1,1]^d$, and (2) targets sampled on a unit hypersphere with queries located within a $2\times$ bounding box. As shown in Figure~\ref{fig:dimension}, query time increases with dimension for all methods. In the uniform setting, our approach demonstrates comparable performance at lower dimensions, though its query time increases relative to the best baselines as dimensionality grows. In the hypersphere setting, however, it maintains the best performance across all tested dimensions.

\begin{table}[htbp]
\caption{Average query time ($\mu s$) comparison under different insertion orders. Ours: CGAL's \texttt{spatial\_sort} function. Ours-FPO: farthest point ordering.}
\centering\small
\resizebox{0.99\columnwidth}{!}{
  \renewcommand{\arraystretch}{1.25}
  \setlength{\tabcolsep}{7pt}
  \rowcolors{2}{white}{rowgray}
  \begin{tabular}{c|ccccccc}
  \hline
  \rowcolor{headerblue} & Camel & Bunny & Dragon & Kitten & Armadillo & Lucy & Sponza \\
  \hline
  Vertices & 28934 & 72911 & 100313 & 291023 & 726367 & 1018219 & 1313504 \\\hline
  Ours     & 1.704 & 1.718 & 1.797  & 1.789  & 1.972  & 2.264   & 2.843   \\
  Ours-FPO    & 2.043 & 2.060 & 2.249  & 2.176  & 2.539  & 2.995   & 3.293   \\
  \hline
  \end{tabular}
}
\label{table:insertion_order}
\end{table}

\paragraph{Impact of Insertion Order.}
Point insertion order affects both construction and query efficiency. \citet{11165079} explored two sorting strategies---CGAL's \texttt{spatial\_sort} function and farthest point ordering (FPO)---finding that FPO yields slower preprocessing but slightly faster $1$-NN queries. We evaluate both strategies for $k$-NN performance. Table~\ref{table:insertion_order} compares average query times using CGAL's \texttt{spatial\_sort} (Ours) versus FPO (Ours-FPO). Contrary to the $1$-NN case, FPO results in slower $k$-NN queries. This stems from two factors: first, while FPO accelerates the initial $1$-NN lookup, this step constitutes only a small fraction of the total $k$-NN query time; second, the strategy promotes a uniform distribution of inserted points, causing successor lists to include spatially distant points that weaken the locality essential for efficient $k$-NN traversal.

\section{Limitations, Future Work and Conclusion}

In this paper, we present a comprehensive extension of the dynamic programming-based nearest neighbor framework, enabling exact $k$-NN queries on manifold point clouds. By uncovering the recursive dependencies encoded within the construction history, we derive a theoretically grounded algorithm that achieves $1\times$--$10\times$ speedups over state-of-the-art spatial indexing structures in volume-to-surface query scenarios. We further augment the framework with two additional capabilities: prefix queries that facilitate $k$-NN searches within any subset $\{p_1, \ldots, p_m\}$ with zero overhead, and point deletion via local Delaunay updates, thereby establishing a fully dynamic point set data structure.

Despite these advantages, our approach has certain limitations. First, the reliance on incremental Delaunay triangulation and Successor Table construction incurs higher preprocessing times, rendering the method best suited for query-intensive applications where this initial overhead can be effectively amortized. Second, when both query and target points reside on the same manifold surface---as in surface normal estimation---our method experiences a slight performance degradation compared to highly optimized, domain-specific alternatives. Finally, the current point deletion mechanism does not yet match the efficiency of specialized dynamic structures, making it less ideal for environments demanding high-frequency updates.

Addressing these limitations presents compelling avenues for future research. Algorithmically, we aim to accelerate preprocessing by parallelizing the insertion of spatially distant points, akin to strategies employed in CGAL. Additionally, exploring lazy update mechanisms and batch reordering could significantly enhance deletion efficiency and cache locality. Beyond software optimizations, GPU acceleration offers a promising yet challenging frontier. While individual $k$-NN queries are inherently parallelizable, the sequential dependencies of Delaunay construction, the need for contiguous memory layouts to ensure coalescing, and the fine-grained synchronization required for concurrent deletions pose significant hurdles. We are actively exploring these interconnected directions to fully realize our framework's potential.

\begin{acks}
The authors thank the anonymous reviewers for their insightful comments and suggestions. This work was supported by the National Natural Science Foundation of China (Grants U23A20312, 62272277, and 62472257) and the Natural Science Foundation of Shandong Province (Grant ZR2025MS986).
\end{acks}

\bibliographystyle{ACM-Reference-Format}
\bibliography{main}

\end{document}